\documentclass[12pt]{article}
\usepackage{multirow}
\usepackage[T1]{fontenc}
\usepackage{algpseudocode}
\usepackage{algorithm}

\usepackage{graphicx}
\usepackage{enumerate}
\usepackage{rotating}
\usepackage{tikz}
\usepackage{amsmath, amsthm, amsfonts, amssymb}

\newtheorem{theorem}{Theorem}
\newtheorem{lemma}[theorem]{Lemma}
\newtheorem{corollary}[theorem]{Corollary}

\newtheorem{proposition}[theorem]{Proposition}
\newtheorem{conjecture}{Conjecture}
\newtheorem{definition}{Definition}

\begin{document}
\title{\bf Equitable total coloring of corona of cubic graphs}
\author{
Hanna Furma\'{n}czyk \\
\small{\emph{Institute of Informatics, Faculty of Mathematics, 
Physics and Informatics,}} \\
\small{\emph{University of Gda\'{n}sk, Wita Stwosza 57, 80-308 Gda\'{n}sk, Poland}}\\
\small{\emph{hanna.furmanczyk@inf.ug.edu.pl}}\\
\\ 
Rita Zuazua\\
\small{\emph{Department of Mathematics, Faculty of Sciences}}\\
\small{\emph{National Autonomous University of Mexico}}\\
\small{\emph{Ciudad Universitaria, Coyoacan, 04510 Mexico, DF, Mexico}}\\
\small{\emph{ritazuazua@ciencias.unam.mx}}
}
\date{}

\maketitle
\begin{abstract}
The minimum number of total independent partition sets of $V \cup E$ of a graph $G=(V,E)$ is called the \emph{total chromatic number} of $G$, denoted by $\chi''(G)$. If the difference between cardinalities of any
two total independent sets is at most one, then the minimum number of total independent partition sets of $V \cup E$ is called the \emph{equitable total chromatic number}, and is denoted by $\chi''_=(G)$.
 
In this paper we consider equitable total coloring of coronas of cubic graphs, $G \circ H$. It turns out that, independly on the values of equitable total 
chromatic number of factors $G$ and $H$, equitable total chromatic number of corona $G \circ H$ is equal to $\Delta(G \circ H) +1$. Thereby, we confirm Total Coloring 
Conjecture (TCC), posed by Behzad in 1964,
and Equitable Total Coloring Conjecture (ETCC), posed by Wang in 2002, for coronas of cubic graphs.
As a direct consequence we get that all coronas of cubic graphs are of Type 1.
\end{abstract}



\section{Introduction}
Graph coloring is one of the most important problems in graph theory. As an extension of proper vertex and edge coloring, the concept of total coloring is developed. In the paper we consider one 
of non-classical models of total coloring, namely equitable total coloring.  

A \emph{$k$-total-coloring} of $G$ is an assignment of $k$ colors to the
edges and vertices of $G$, so that adjacent or incident elements obtain
different colors. The \emph{total chromatic number} of $G$, denoted by
$\chi''(G)$, is the smallest $k$ for which $G$ has a
$k$-total-coloring. Clearly, $\chi''(G)\geq \Delta(G) +1$, where $\Delta(G)$ is the maximum degree of $G$.  Well known Total Coloring Conjecture~\cite{behzad,Vizing} states that the total 
chromatic number of 
any graph is at most $\Delta(G)+2$. 

\begin{conjecture}{\emph{[TCC]}}{\emph{\cite{behzad,Vizing}}}
For any graph $G$ the following inequalities hold
$$ \Delta(G)+1 \leq \chi''(G) \leq \Delta(G)+2.$$
\end{conjecture}

Although the hypothesis has been known since 1964, it has been proven only for some specific classes of graphs, in particular for cubic graphs~\cite{vijayaditya1971}. Graphs with $\chi''(G) = \Delta(G)+1$ are said to be \emph{Type~1}, 
and graphs with $\chi''(G) = \Delta(G)+2$ are said to be \emph{Type~2}. The problem of deciding whether a graph is Type~1 has been shown to be NP-complete even for cubic bipartite 
graphs~\cite{Sanchez}.

In this paper one of non-classical models of total coloring is considered. A $k$-total-coloring is \emph{equitable} if the cardinalities of any 
two color classes differ by at most one (ref. Fig.\ref{fig:k33}). The smallest $k$ for which $G$ has an equitable $k$-total-coloring is the \emph{equitable total 
chromatic number} of $G$, and it is denoted by $\chi''_{=}(G)$. The concept of equitable total coloring was first presented in \cite{fu}. 
This model of graph coloring has many practical applications. Every time when we have to divide a system with binary conflict relations into equal or almost equal 
conflict-free subsystems we can model this situation by means of equitable graph coloring. In particular, one motivation for equitable coloring suggested by Meyer 
\cite{meyer} concerns scheduling problems. Furma\'nczyk \cite{furm} mentions a specific application of this type of scheduling problem, namely, assigning university 
courses to time slots in a way that avoids scheduling incompatible courses at the same time and spreads the courses evenly among the available time slots. The topic of equitable coloring, also
its total version, was widely discussed in literature. Similarly to the situation with proper total coloring, it was conjectured that the equitable total chromatic number of any graph is at most $\Delta(G)+2$.
\begin{conjecture}{\emph{[ETCC]}}{\emph{\cite{wang}}}
For any graph $G$ the following inequalities hold
$$ \Delta(G)+1 \leq \chi''_=(G) \leq \Delta(G)+2.$$
\end{conjecture}
This conjecture was proven among others for cubic graphs in~\cite{wang}. 
Wang \cite{wang} proved that every cubic graph has an equitable total coloring with 5 colors. Recently, it has been shown that the problem of deciding whether the equitable total chromatic number of a bipartite cubic graph is~4 is NP-complete~\cite{DianaICGT}.

\begin{figure}
    \centering
    \includegraphics[scale=0.5]{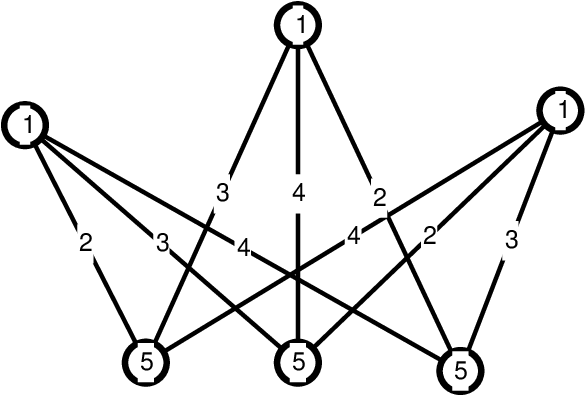}
    \caption{An exemplary equitable total 5-coloring of $K_{3,3}$.}
    \label{fig:k33}
\end{figure}

One can ask whether there exist graphs with equitable total chromatic number greater than total chromatic number. It turns out the answer to this question is positive. There are known 
examples of cubic graphs such that their total chromatic number is strictly less than their equitable total chromatic 
number~\cite{DianaICGT, fu}. 

In this paper we ask about the value of the equitable total coloring number of graph products. The problem was considered for some Cartesian products of graphs \cite{chunling}. Moreover, graph products are interesting
and useful in many situations. The complexity of many problems, also equitable coloring, that deal with very large and complicated graphs is reduced greatly if one is able to fully characterize the properties of less complicated prime factors. We continue the research on graph products, but this time as a factor we take cubic graphs and we consider corona product of graphs. 

Given two simple graphs $G$ and $H$, the \emph{corona product of $G$ and $H$} is the graph $G \circ H$ obtained by taking one copy of $G$, $\left|V(G)\right|$ copies 
of $H$, and making the $i$\textsuperscript{th} vertex of $G$ adjacent to every vertex of the $i$\textsuperscript{th} copy of $H$, $H_i$ (ref. Fig.~\ref{cubical}). 
This graph product was introduced by Frucht and Harary in 1970 \cite{frucht}.

\begin{figure}[htb]
\begin{center}
\begin{tikzpicture}[scale=1]
	\coordinate (LU) at (-0.5, 0.5); \coordinate (LD) at (-0.5, -0.5);
	\coordinate (RU) at (0.5, 0.5); \coordinate (RD) at (0.5, -0.5);
	\coordinate (LU1) at (-2+0.6, 2+0.6);
	\coordinate (LU2) at (-2+0.2, 2+0.2);
	\coordinate (LU3) at (-2-0.2, 2-0.2);
	\coordinate (LU4) at (-2-0.6, 2-0.6);
	\coordinate (LD1) at (-2+0.6,-2-0.6);
	\coordinate (LD2) at (-2+0.2,-2-0.2);
	\coordinate (LD3) at (-2-0.2,-2+0.2);
	\coordinate (LD4) at (-2-0.6,-2+0.6);
	\coordinate (RU1) at ( 2+0.6, 2-0.6);
	\coordinate (RU2) at ( 2+0.2, 2-0.2);
	\coordinate (RU3) at ( 2-0.2, 2+0.2);
	\coordinate (RU4) at ( 2-0.6, 2+0.6);
	\coordinate (RD1) at ( 2+0.6,-2+0.6);
	\coordinate (RD2) at ( 2+0.2,-2+0.2);
	\coordinate (RD3) at ( 2-0.2,-2-0.2);
	\coordinate (RD4) at ( 2-0.6,-2-0.6);

	\foreach \i in {LU,LD,RU,RD,LU1,LU2,LU3,LU4,LD1,LD2,LD3,LD4,RU1,RU2,RU3,RU4,RD1,RD2,RD3,RD4}%
		\fill (\i) circle(2pt);
		
	\foreach \a/\b in {LU/LD,LU/RU,LU/RD,LD/RU,LD/RD,RU/RD} \draw (\a)--(\b);
	\foreach \a/\b in {LU/LU1,LU/LU2,LU/LU3,LU/LU4,LD/LD1,LD/LD2,LD/LD3,LD/LD4,%
		RU/RU1,RU/RU2,RU/RU3,RU/RU4,RD/RD1,RD/RD2,RD/RD3,RD/RD4}%
		\draw (\a)--(\b);
	\foreach \a/\b in {LU1/LU2,LU2/LU3,LU3/LU4,LD1/LD2,LD2/LD3,LD3/LD4,%
		RU1/RU2,RU2/RU3,RU3/RU4,RD1/RD2,RD2/RD3,RD3/RD4} \draw (\a)--(\b);
	\draw (LU1) .. controls +(-0.4,0) and +(0,+0.4) .. (LU3);
	\draw (LU2) .. controls +(-0.4,0) and +(0,+0.4) .. (LU4);
	\draw (LD1) .. controls +(-0.4,0) and +(0,-0.4) ..(LD3);
	\draw (LD2) .. controls +(-0.4,0) and +(0,-0.4) ..(LD4);
	\draw (RU3) .. controls +(0.4,0) and +(0,0.4) ..(RU1);
	\draw (RU4) .. controls +(0.4,0) and +(0,0.4) ..(RU2);
	\draw (RD3) .. controls +(0.4,0) and +(0,-0.4) ..(RD1);
	\draw (RD4) .. controls +(0.4,0) and +(0,-0.4) ..(RD2);
	
	\draw (LU1) ..controls +(-0.8,0.2) and +(-0.2,0.8) ..(LU4);
	\draw (LD1) ..controls +(-0.8,-0.2) and +(-0.2,-0.8) ..(LD4);
	\draw (RU1) ..controls +(0.2,0.8) and +(0.8,0.2) ..(RU4);
	\draw (RD1) ..controls +(0.2,-0.8) and +(0.8,-0.2) ..(RD4);
\end{tikzpicture}
\end{center}
\caption{Corona $K_4 \circ K_4$.}
\label{cubical}
\end{figure}
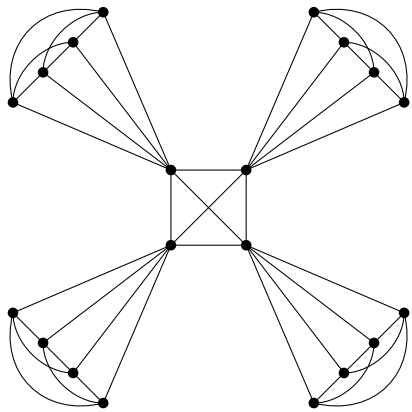

In this paper we focus on coronas of two arbitrary cubic graphs. This kind of graph product seems to be interesting because corona graphs lie often close to the boundary between easy and hard problems \cite{ars}.
Here, we ask whether the fact of being the cubic graph of Type 1 or 2 has the influence on the value of the equitable total chromatic number of the corona of such factors. It turns out that the answer is negative.
Let $G$ and $H$ be two cubic graphs with $\left|V(G)\right|=n_G$ and $\left|V(H)\right|=n_H$ vertices, 
respectively. It is easy to see that the 
maximum degree of the corona graph $G\circ H$ is $\Delta(G\circ H)=n_H+3$. We prove that (i) the total chromatic number of $G\circ H$ is equal to 
$\Delta(G\circ H)+1=n_H+4$, i.e, they are all of Type~1, and (ii) the equitable total chromatic number of the corona graph $G\circ H$ is equal to 
$\Delta(G\circ H)+1=n_H+4$, independently of the type of factors $G$ and~$H$.

\section{Notation and definitions}

In the paper we will use the concept of semi-graphs, introduced by Dantas et al. \cite{DianaICGT}.

\begin{definition}
A \emph{semi-graph} is a triple $G=(V, E, S)$, where $V(G)$ is a set of vertices of $G$, $E(G)$ is a set of edges having two distinct endpoints in $V(G)$, and $S(G)$
is a multiset of \emph{semi-edges} having one endpoint in $V(G)$.
\end{definition}

Note that if $S(G)=\emptyset$ then a semi-graph $G$ is a simple graph. All definitions given below for semi-graphs, that do not require the existence of semi-edges, are also valid for graphs. When it could be confusing we explicitly write graph or semi-graph.
We write edges having endpoints $v$ and $w$ shortly as $vw$ and semi-edges having endpoint $v$ as $v\cdot$. When vertex $v$ is an endpoint of $e\in E \cup S$ we say that $v$ and $e$ are \emph{incident}. Two elements of $E\cup S$ incident to the same vertex, respectively two vertices incident to the same edge, are called \emph{adjacent}. $N(v)$ denotes the \emph{open neighborhood} of a vertex $v \in V(G)$, i.e. the set of adjacent vertices for $v$. $N[v]=N(v)\cup \{v\}$ is the \emph{close neighborhood} of $v$.
The \emph{degree} $\deg(v)$ of a vertex $v$ of $G$ is the number of elements of $E\cup S$ that are incident to $v$. We say that $G$ is $r$-\emph{regular} if the degree of each vertex is equal to $r$.
An exemplary semi-graph $G$ is given in Fig.~\ref{semicorona}, where $V(G)=\{a,b,c,d\}$, $E(G)=\{ab, ac, ad, bc, bd, cd\}$, and $S(G)=\{a\cdot,a\cdot, a\cdot, b\cdot, b\cdot,$ $b\cdot, c\cdot,$ $c\cdot, c\cdot, d\cdot, d\cdot, d\cdot\}$.

For a given graph $G=(V(G),E(G))$, where $V(G)=\{v_1,\ldots,v_{n_G}\}$,  and graph $H=(V(H),E(H))$ with $V(H)=\{u_1,\ldots,u_{n_H}\}$, $|V(H)|=n_H$, for any $i \in \{1,\ldots,n_G\}$ we define an \emph{open fan} of $v_i\in V(G)$ as a set of $n_H$ semi-edges with common endvertex $v_i$ and we denote it by $F_H(v_i)$. A \emph{close fan} $F_H[v_i]$ is a set $F_H(v_i) \cup \{v_i\}$. For any $j \in \{1,\ldots,n_H\}$, we define the \emph{open claw} of $u_j\in V(H)$ as a set of edges in $H$ incident to $u_j$, and we denote it by $I_H(u_j)$. We have $I_H(u_j) \subset E(H)$. A \emph{close claw} of $u_j$, $I_H[u_j]$, is a set $I_H(u_j)\cup \{u_j\}$.

\begin{definition}
A \emph{semi-corona} $G\circ_s H$ of a graph $G=(V(G),E(G))$ and a graph $H=(V(H),E(H))$ is the semi-graph $G'=(V(G'),E(G'),S(G'))$, where 
$$V(G')=V(G),$$ $$E(G')=E(G),$$ $$S(G')=\bigcup_{v \in V(G)} F_H(v).$$ 
\end{definition}

A semi-corona $G\circ_s H$ may be also defined as the semi-graph obtained from graph $G$ by adding $n_H$ semi-edges to each vertex of $G$. 
It is easy to see that semi-corona $G \circ_s H$ of a cubic graph $G$ and $n_H$-vertex cubic graph $H$ is $(n_H+3)$-regular semi-graph. An example of semi-corona is given in Fig.~\ref{semicorona}. 

\begin{figure}[htb]
\setlength{\unitlength}{1pt}
\begin{center}
\begin{picture} (80,80) 
\put(0,25){\line(1,0){20}} 
\put(20,25){\line(-1,-1){16.6}}
\put(20,5){\line(0,1){20}}
\put(0,65){\line(1,0){20}} \put(20,65){\line(0,1){20}}
\put(60,25){\line(1,0){20}} 
\put(60,25){\line(1,-1){16.6}}
\put(60,5){\line(0,1){20}}
\put(60,65){\line(1,0){20}} \put(60,65){\line(0,1){20}}

\put(20,25){\circle*4} \put(23,15){$a$}
\put(20,65){\circle*4} \put(23,67){$d$}
\put(60,25){\circle*4} \put(50,15){$b$}
\put(60,65){\circle*4} \put(50,67){$c$}

\put(20,25){\line(1,0){40}}
\put(20,25){\line(0,1){40}}
\put(20,25){\line(1,1){56.6}}
\put(60,25){\line(-1,1){56.6}}
\put(60,25){\line(0,1){40}}
\put(20,65){\line(1,0){40}}

\end{picture}
\caption{Semi-corona $K_4 \circ_s H$, where $H$ is a 3-vertex graph.}\label{semicorona}
\end{center}
\end{figure}
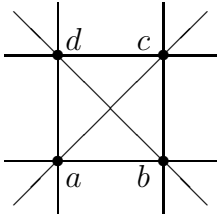

Now, we will define the operation $+_{v_i}$. For a given semi-graph $G \circ_s H=G_0=(V(G_0),E(G_0),$ $S(G_0))$ with $V(G)=\{v_1,\ldots,v_{n_G}\}$ and graph $H=(V(H),E(H))$, we define $G_i=(V(G_i),E(G_i),S(G_i))$ as a semi-graph $G_{i-1}+_{v_i} H$, where $$V(G_i)=V(G_{i-1})\cup V(H_i),$$ $$E(G_i)=E(G_{i-1})\cup E(H_i)\cup \{v_iw: w\in V(H_i)\},$$ $$S(G_i)=S(G_{i-1}) \backslash F(v_i).$$
It is easy to see that $G_{n_G}=G \circ H$. Of course, $H_i$ denotes the $i$-th copy of $H$.
We will name graphs $G_1,\ldots,G_{n_G}$ as \emph{extended semi-coronas} (ref. Fig.~\ref{ext}).

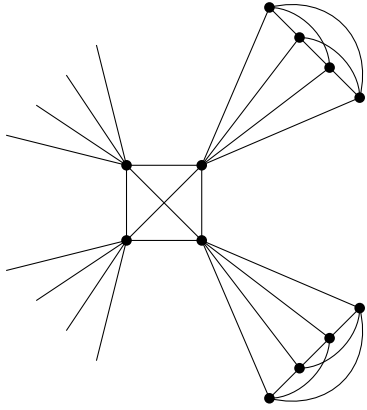
\begin{figure}[htb]
\begin{center}
\begin{tikzpicture}[scale=1]
	\coordinate (LU) at (-0.5, 0.5);
	\coordinate (LD) at (-0.5, -0.5);
	\coordinate (RU) at (0.5, 0.5); \coordinate (RD) at (0.5, -0.5);
	\coordinate (LU1) at (-1.5+0.6, 1.5+0.6);
	\coordinate (LU2) at (-1.5+0.2, 1.5+0.2);
	\coordinate (LU3) at (-1.5-0.2, 1.5-0.2);
	\coordinate (LU4) at (-1.5-0.6, 1.5-0.6);
	\coordinate (LD1) at (-1.5+0.6,-1.5-0.6);
	\coordinate (LD2) at (-1.5+0.2,-1.5-0.2);
	\coordinate (LD3) at (-1.5-0.2,-1.5+0.2);
	\coordinate (LD4) at (-1.5-0.6,-1.5+0.6);
	\coordinate (RU1) at ( 2+0.6, 2-0.6);
	\coordinate (RU2) at ( 2+0.2, 2-0.2);
	\coordinate (RU3) at ( 2-0.2, 2+0.2);
	\coordinate (RU4) at ( 2-0.6, 2+0.6);
	\coordinate (RD1) at ( 2+0.6,-2+0.6);
	\coordinate (RD2) at ( 2+0.2,-2+0.2);
	\coordinate (RD3) at ( 2-0.2,-2-0.2);
	\coordinate (RD4) at ( 2-0.6,-2-0.6);

	\foreach \i in {LU,LD,RU,RD,RU1,RU2,RU3,RU4,RD1,RD2,RD3,RD4}%
		\fill (\i) circle(2pt);
		
	\foreach \a/\b in {LU/LD,LU/RU,LU/RD,LD/RU,LD/RD,RU/RD} \draw (\a)--(\b);
	\foreach \a/\b in {LU/LU1,LU/LU2,LU/LU3,LU/LU4,LD/LD1,LD/LD2,LD/LD3,LD/LD4,%
		RU/RU1,RU/RU2,RU/RU3,RU/RU4,RD/RD1,RD/RD2,RD/RD3,RD/RD4}%
		\draw (\a)--(\b);
	\foreach \a/\b in {
		RU1/RU2,RU2/RU3,RU3/RU4,RD1/RD2,RD2/RD3,RD3/RD4} \draw (\a)--(\b);
	\draw (RU3) .. controls +(0.4,0) and +(0,0.4) ..(RU1);
	\draw (RU4) .. controls +(0.4,0) and +(0,0.4) ..(RU2);
	\draw (RD3) .. controls +(0.4,0) and +(0,-0.4) ..(RD1);
	\draw (RD4) .. controls +(0.4,0) and +(0,-0.4) ..(RD2);
	
	\draw (RU1) ..controls +(0.2,0.8) and +(0.8,0.2) ..(RU4);
	\draw (RD1) ..controls +(0.2,-0.8) and +(0.8,-0.2) ..(RD4);
\end{tikzpicture}
\end{center}
\caption{Extended semi-corona $G_2=((K_4 \circ_s K_4) +_{v_1} K_4)+_{v_2} K_4$.}
\label{ext}
\end{figure}

For $k \in \mathbb{N}^{+}$ and given semi-graph $G=(V,E,S)$, a \emph{proper vertex $k$-coloring} of $G$ is a map $c_V: V \rightarrow \{1,\ldots,k\}$ such that $c_V(x) \neq c_V(y)$ for any two adjacent vertices $x$ and $y$. The smallest number of colors admitting such a coloring is named as the \emph{chromatic number} and it is denoted by $\chi(G)$.

Similarly, a \emph{proper edge $k$-coloring} of $G$ is a map $c_{E\cup S} : E \cup S \rightarrow \{1,\ldots,k\}$ such that $c_{E\cup S}(e_1) \neq c_{E\cup S}(e_2)$ for any two adjacent elements $e_1, e_2$ of $E \cup S$. If $S=\emptyset$ then we will write $c_{E} : E \rightarrow \{1,\ldots,k\}$. The smallest number of colors admitting such a coloring is named as the \emph{chromatic index} and it is denoted by $\chi'(G)$.

A \emph{total $k$-coloring} of $G$ is a map $c_T:V\cup E \cup S \rightarrow \{1,\ldots,k\}$ such that
\begin{itemize}
    \item $\left. c_T \right |_V$ is a proper vertex coloring,
    \item $\left. c_T \right |_{E\cup S}$ is a proper edge coloring,
    \item $c_T(e)\neq c_T(v)$ whenever $e\in E \cup S$, $v\in V$ and $e$ is incident to $v$.
\end{itemize}

A vertex (edge (total)) $k$-coloring is \emph{equitable} if the cardinalities of any two color classes differ by at most one.  

For a given vertex (edge (total)) $k$-coloring of a graph $G$, this means for a partition of the appropriate set into $k$ independent color classes $\{P_1,P_2, \ldots, P_k\}$, 
the \emph{vertex $($edge $($total$))$ coloring sequence} $S_V(G)$ ($S_E(G)$ ($S_T(G)$)) is a sequence of their cardinalities, i.e. $(|P_1|,|P_2|, \ldots, |P_k|)$. For the total 5-coloring of $K_{3,3}$ given in Figure \ref{fig:k33} $S_T(K_{3,3})=(3,3,3,3,3)$. For the vertex coloring being restriction of the total coloring to $V$ the sequence $S_V(K_{3,3})=(3,0,0,0,3)$. Similarly we get $S_E(K_{3,3})=(0,3,3,3,0)$.

\section{Equitable coloring of cubic graphs}
Let us remind some known results concerning coloring of cubic graphs that will be useful in the further part of this work. First of all, let us notice that, when we consider only vertex coloring, the chromatic 
number is equal to the equitable chromatic number for all connected 
cubic graphs~\cite{Chen}. This means that every proper vertex coloring of connected cubic graph $G$ with $\chi(G)$ colors can be made equitable without adding new colors. 
\begin{theorem}[\cite{Chen}]
If $G$ is a connected cubic graph then $$\chi(G) = \chi_=(G).$$
\end{theorem}

\begin{corollary}
If $G$ is a connected cubic graph then $$2 \leq \chi_=(G) \leq 4.$$ \hfill $\Box$
\end{corollary}

In the case of equitable edge coloring, it is known, for example from~\cite{ZhangWang}, that the equitable chromatic index for any graph is equal to its chromatic index. 

\begin{theorem}[\cite{ZhangWang}]
Let $G$ be a simple graph. Then $$\chi_='(G) = \chi'(G).$$
\end{theorem}

\begin{theorem}[\cite{ZhangWang}]
Every graph $G$ has an equitable edge $k$-coloring for each $k \geq \chi_='(G)$.\label{contedge}
\end{theorem}

Let us recall also Vizing theorem.

\begin{theorem}[\cite{Vizing}]
Let $G$ be a graph. Then $$\Delta(G) \leq \chi'(G) \leq \Delta(G)+1.$$\label{viz}
\end{theorem}

For more information about equitable vertex and edge colorings we refer to~\cite{Hanna}. 

As we have already mentioned, cubic graphs are one of graph classes that the Equitable Total Coloring Conjecture holds for. We have

\begin{theorem}[\cite{wang,subc}]\label{thm:cont}
Every cubic graph $G$ can be equitably total colored with $k$ colors for every $k \geq 5$.
\end{theorem}

In the further part of our paper we will use equitable total $(n+4)$-coloring of a $n$-vertex cubic graphs $G$. Such a coloring exists due to Theorem~\ref{thm:cont}. Now, we give some properties of such a coloring.

\begin{proposition}
In any equitable total $(n+4)$-coloring of $n$-vertex cubic graph $G$ the cardinalities of color classes are between $1$ and $3$.\label{clca}
\end{proposition} 
\begin{proof}
By the contrary, let us assume that there is at least one color class of the cardinality at least 4. Since the coloring is equitable, the remaining color classes are of cardinality at least 3. This means that the number of elements (vertices and edges) 
in $G$ is not less than $4+3(n+3)=3n+13$, while we know that this number is equal to $5/2n$. Contradiction.

On the other hand, since $5n/2 > n +4$, the cardinalities of color classes are obviously greater or equal to 1.
\end{proof}

Let $S_T(G)$ be the total coloring sequence of an equitable total $(n+4)$-coloring of $n$-vertex cubic graph $G$.
Let $\#_j(S_T(G))$ denote the number of terms (color classes) in $S_T(G)$ of cardinality $j$, $j=1,2,3$.

\begin{proposition}\label{totalsequence}
Let $G$ be a $n$-vertex cubic graph colored in an equitable total way with $n+4$ colors.

\begin{enumerate}
\item[\emph{($i$)}] If $4 \leq n \leq 14$ then:
\begin{align*}
\#_1(S_T(G)) &=8-n/2,\\
\#_2(S_T(G)) &= 3n/2 -4,\\
\#_3(S_T(G)) &=0.
\end{align*}

\item[\emph{($ii$)}] If $n \geq 16$ then:
\begin{align*}
\#_1(S_T(G)) &=0,\\
\#_2(S_T(G)) &= n/2 +12,\\
\#_3(S_T(G)) &=n/2 -8.
\end{align*}
\end{enumerate}\label{termS_T}
\end{proposition}
\begin{proof}
It is easy to observe that the only value of $n$ when all terms in $S_T(G)$ are equal to 2 is $n=16$. It is enough to solve equation $2(n+4)=5n/2$. For smaller number of vertices, terms 
in a  sequence $S_T(G)$ are equal to 1 and 2, for bigger ones - to 2 and 3. Now, the values of $\#_j(S_T(G))$, $j=1,2,3$, are the solutions of system of equations 
for $j=1$ in Case (i) and for $j=2$ in Case (ii):
$$
\left\{
\begin{array}{ll}
j\cdot\#_j(S_T(G)) +(j+1)\cdot\#_{j+1}(S_T(G)) &=5n/2\\
\#_j(S_T(G))+\#_{j+1}(S_T(G)) &= n +4.
\end{array}\right.
$$
\end{proof}

\section{Equitable total coloring of semi-coronas}

\begin{lemma}
Let $G\circ_s H$ be a semi-corona of cubic graphs: $n_G$-vertex graph $G$ and $n_H$-vertex graph $H$. Then $$\chi_=''(G\circ_s H) = \Delta(G\circ_s H)+1=n_H+4.$$ \label{lmsemicorona}
\end{lemma}
\begin{proof}
Since $\chi_=''(T) \geq \Delta(T)+1$ for any graph $T$, all we need is to 
construct an equitable total $(n_H+4)$-coloring of $G\circ_s H$. Do as follows:
\begin{enumerate}
\item Color equitably edges of cubic graph $G$ with $n_H + 4$ colors in such a way that the corresponding edge color 
sequence $S_E(G)$ is equal to $(l_e(1), l_e(2), \ldots, l_e(n_H+4))=(\lceil (3/2n_G)/(n_H+4) \rceil, \lceil (3/2n_G-1)/(n_H+4) \rceil, \ldots, \lceil (3/2n_G-n_H-3)/(n_H+4) \rceil)$,
where $l_e(i)$ denotes the number of edges in $G$ colored with $i$.

Since $3 \leq \chi'_=(G) \leq 4$ for every cubic graph $G$ and we color edges of $G$ with at least 8 colors, this step is possible due to Theorem \ref{contedge}.

\item Extend this coloring into any proper total $(n_H+4)$-coloring of $G$. \label{krok2}

Let us assume that all edges and some vertices of $G$ have been already colored. Notice that for every uncolored vertex $v \in V(G)$ at most six colors are forbidden - the colors assigned to 
three incident edges and at most three adjacent vertices, if they have already been colored. 
Since we have $n_H+4 \geq 8$ colors, then there are at least two allowed colors for every vertex $v$. We can choose one of them.
Let $S_V(G):=(l_v(1), l_v(2), \ldots, l_v(n_H+4))$ be the corresponding vertex coloring sequence. Of course, $S_E(G)+S_V(G)$ is the total coloring sequence of $G$.

\item Extend the total coloring of $G$ into an equitable total $(n_H+4)$-coloring of semi-corona $G \circ_s H$ by coloring properly semi-edges of $G\circ_s H$, i.e. elements of an open fan $F(v)$ for every $v \in V(G)$.

Note that exactly 4 colors are not allowed to color semi-edges from $F(v)$. Let $c(F(v))$ denote the set of all allowed colors for semi-edges from $F(v)$. Since $|c(F(v))|=n_H$, the coloring of $F(v)$ is determined, with an accuracy to the permutations of $c(F(v))$. 

We claim that the total coloring of $G \circ_s H$ obtained in the way described above is equitable. Indeed, let us notice that color $i$ used to color vertex $v\in V(G)$ implies $i \not\in c(F(v))$ while color $i$ used to color edge $e=uv\in E(G)$ implies $i \notin c(F(u))$ and $i \notin c(F(v))$. 
Thus, the fact that a color $i$ is used to color $l_v(i)$ vertices and $l_e(i)$ edges in $G$ means that the color $i$ will appear in $n_G-l_v(i) - 2l_e(i)$ sets of available colors $c(F(v))$ and this means that it can be used to 
color $n_G-l_v(i) - 2l_e(i)$ semi-edges. Thus, color $i$ is used $l_v(i)+l_e(i)+n_G-l_v(i)-2l_e(i)=n_G-l_e(i)$ times. Since
the sequence $(l_e(1), l_e(2), \ldots,l_e(n_H+4))$
from the first step was equitable, then the sequence $(n_G-l_e(1),\ldots,n_G-l_e(n_H+4))$ is also equitable. Thus
the extended total coloring of $G \circ_s H$ is equitable.
\end{enumerate}
\end{proof}

\section{Main result}

In this section we prove the main theorem of this paper.

\begin{theorem}
Let $G$ and $H$ be cubic graphs on $n_G$ and $n_H$ vertices, respectively. Then
$$\chi_=''(G \circ H) = \Delta(G \circ H) +1 = n_H +4.$$ \label{main}
\end{theorem}
\begin{proof} Let $n_H\geq 6$. For such cases the main idea of the proof is:
\begin{enumerate}
    \item to color semi-corona $G \circ_s H$ in an equitable total way with $n_H+4$ colors; we get the total equitable coloring sequence $S_T(G \circ_s H)$;
    \item to color graph $H_1$ in an equitable total way with $n_H+4$ colors;
    \item to make a permutation of colors in the coloring of $H_1$ in such a way that $S_T(G\circ_s H)+S_T(H_1)$ is an equitable total coloring sequence of the extended semi-corona $G_1$; 
    \item to show that the coloring of $H_1$ may be ''joined'' with the coloring of semi-corona to obtain a proper total coloring of $G_1$.
\end{enumerate}
Given the equitable total coloring of the extended semi-corona $G_i$, $i=1, \ldots, n_G-1$ with the corresponding total coloring sequence $S_T(G_i)$, 
\begin{enumerate}
\setcounter{enumi}{4}
    \item color graph $H_{i+1}$ in an equitable total way with $n_H+4$ colors;
    \item make a permutation of colors in the coloring of $H_{i+1}$ in such a way that $S_T(G_i)+S_T(H_{i+1})$ is an equitable total coloring sequence of the extended semi-corona $G_{i+1}$;
    \item ''join'' the coloring of $H_{i+1}$ with the coloring of $G_i$ to obtain a proper total coloring of $G_{i+1}$.
\end{enumerate}

Since $G_{n_G}=G\circ H$, finally we get an equitable total coloring of the whole corona $G \circ H$. Now, all we need is to clarify the above steps and to
show that they are possible to do.

\begin{description}
\item[Ad Step 1] We color semi-corona $G \circ_s H$ with $n_H+4$ colors due to the way given in the proof of Lemma \ref{lmsemicorona}.
\item[Ad Steps 2 and 5] An equitable total $(n_H+4)$-coloring of copies of $H$ is possible to achieve due to Theorem \ref{thm:cont}. 
\item[Ad Steps 3 and 6] Since in these steps graph $H$ is treated independently on the structure of the appropriate extended semi-corona, such permutation is possible to do.
\item[Ad Steps 4 and 7] Let $V(G)=\{v_1,\ldots,v_{n_G}\}$ and  $V(H)=\{u_1,\ldots,u_{n_H}\}$. To show that the equitable total coloring of a copy of $H$ may be 'joined'' with the coloring of the appropriate semi-graph we will prove that for every close claw $I_{H}[u_j]$, $1 \leq j \leq n_H$ there exists semi-edge $e$ in $F_H(v_i)$ colored with $c(e)$ such that if we assign the same color $c(e)$ to the edge $v_iu_j$ then we get proper (partial) total coloring of $G \circ H$, $1\leq i \leq n_G$. 

We order colors in $c(F_H(v_i))$ in such a way that six last terms of this order denote colors of cardinality at most 2 in $S_T(H_i)$. Due to 
Proposition \ref{totalsequence} it is possible to do for any $n_H \geq 6$. We get a sequence of colors $(x_1,\ldots, x_H)$. There are at least $n_H-6$ colors in $c(F_H(v_i))$ that can be used to change one semi-edge $e=v_i\cdot$ into an edge $v_iu_1 \in E(G\circ H)$ and to get a proper total coloring. We repeat this for vertices $u_2,\ldots,u_{n_H-6}$. Finally, we have six ''unjoined'' vertices $u_{n_H-5},\ldots, u_{n_H}$ in $H$ and six unassigned semi-edges $e_1,\ldots, e_6$ in the appropriate open fan $F_H(v_i)$. Let $X_k$ denote the set of colors out of $\{c(e_1),\ldots,c(e_6)\}$ that can be assigned to an edge $v_iu_{k}$, $n_H-5\leq k\leq n_H$. We get the family $\mathcal{X}=\{X_{n_H-5}, \ldots, X_{n_H}\}$. 
In the equitable coloring of $H$ the cardinalities of color classes corresponding to the colors of the semi-edges are at most two, $|P_{c(e_i)}| \leq 2$, $1 \leq i \leq 6$. Thus, only at most 12 elements out of 24 in $\bigcup_{n_H-5 \leq j \leq n_H} I_H[u_j]$ can be colored with colors $c(e_1),\ldots,c(e_6)$. So, each color belongs to at least four sets $X_i$ and $|X_i|\geq 2$. This means that for each subset of indexes $K$ of a family $\mathcal{X}$ we get $|\bigcup_{k\in K}X_k| \geq |K|$. Thus, there exists a transversal for $\mathcal{X}$ by Hall's marriage theorem \cite{hall}. In other words, we are able to select one representative for each set $X_i$ in such a way that no two sets from $\mathcal{X}$ get the same representative.  
Thus, the equitable total coloring of a copy of $H$ may be 'joined'' with the coloring of the appropriate semi-graph, if only $n_H \geq 6$. 

\end{description}

Now, we need to prove only the correctness of the theorem for $n_H=4$. The algorithm of the equitable total coloring of $G \circ H$ where $n_H=4$ is as follows:
\begin{enumerate}
    \item color semi-corona $G \circ_s K_4$ in an equitable total way with $8$ colors due to the algorithm given in the proof of Lemma \ref{lmsemicorona}; we get the total equitable coloring sequence $S_T(G \circ_s K_4)$;
    
    \item determine an equitable total coloring sequence of a length 8 for $H_1$: $S_T(H_1)$, such that $S_T(G \circ_s K_4)+S_T(H_1)$ results in an equitable total coloring sequence of the extended semi-corona $G_1$;
    \item transform colored $G \circ_s K_4$ into the partially colored extended semi-corona $G_1$ in such a way that colors of vertices and edges of $G$ are not changed, while the colors of $F_H(v_1)$ are assigned arbitrarily to the edges joining $H_1$ with $G$ in $G_1$,
    \item color $H_1$ due to $S_T(H_1)$. Since only two its terms (colors) are of value 2, we color $H_1=K_4$ in such a way that color of cardinality 2 are assigned to one vertex and one edge in $H_1$. The rest of colors are used only once in the coloring of $H_1$. It is easy to verify that it leads us to the proper equitable total coloring of $G_1$.
\end{enumerate}

We generalize Steps 2--4 for next copies of $K_4$ and execute them until we get an equitable total 8-coloring of the whole corona $G \circ K_4$.
\end{proof}

\section{Final remarks}
Since in the proof of Lemma \ref{lmsemicorona} we did not use the fact that $G$ is cubic, we may generalize the lemma to the following

\begin{corollary}
Let $G$ be an $r$-regular graph and let $H$ be a cubic graph on $n_H$ vertices where $n_H\geq r-3$. Then
$$\chi_=''(G\circ_s H) = \Delta(G\circ_s H)+1=n_H+4.$$
\end{corollary}

\noindent Finally we get
\begin{corollary}
Let $G$ be an $r$-regular graph and let $H$ be a cubic graph on $n_H$ vertices where $n_H\geq r-3$.
Then
$$\chi_=''(G \circ H) = \Delta(G \circ H) +1 = n_H +4.$$ 
\end{corollary}

Many interesting questions remain still open, for instance the equitable total colorability of coronas of $r$-regular graphs with $r>3$. We hope that our paper will be a source of inspiration to answer this question.

\end{document}